\documentclass[conference]{IEEEtran}
\IEEEoverridecommandlockouts
\usepackage{textcase}
\usepackage{amsmath,amssymb,amsfonts}
\usepackage{textcomp}
\usepackage{balance}
\usepackage{amsthm}
\usepackage[expansion=true,protrusion=true,final]{microtype}
\newtheorem*{proposition}{Proposition}
\usepackage{graphicx}
\usepackage{booktabs}
\usepackage{siunitx}
\usepackage{cite}
\usepackage{url}
\usepackage{multirow}
\usepackage{makecell}
\usepackage{pgfplots}
\usetikzlibrary{patterns}
\usepackage{array}
\usepackage{quantikz}
\usepackage{subcaption}
\usepackage{pgfplotstable}
\pgfplotsset{compat=1.18}

\usepgfplotslibrary{groupplots}
\usepackage{xcolor}
\usepackage{bm}
\usepackage{titlesec}
\titlespacing*{\section}{0pt}{5pt}{2pt}
\titlespacing*{\subsection}{0pt}{3pt}{1pt}
\titlespacing*{\subsubsection}{0pt}{2pt}{1pt}

\definecolor{cBlue}{HTML}{5C8FBD}
\definecolor{cGreen}{HTML}{5B9E82}
\definecolor{cPurple}{HTML}{8B7AB5}

\begin{document}
\setlength{\textfloatsep}{6pt plus 2pt minus 2pt}
\setlength{\floatsep}{5pt plus 1pt minus 1pt}
\setlength{\intextsep}{5pt plus 1pt minus 1pt}
\setlength{\dbltextfloatsep}{6pt plus 2pt minus 2pt}
\setcounter{secnumdepth}{2}

\title{Configurable Algorithms for Histopathologic Cancer Detection on Quantum Hardware}

\author{%
  \IEEEauthorblockN{Nandika Goyal}
  \IEEEauthorblockA{\textit{SenSIP Center, School of ECEE}\\
  \textit{Arizona State University}\\
Tempe, AZ, USA}
  \and
  \IEEEauthorblockN{Glen Uehara}
  \IEEEauthorblockA{\textit{SenSIP Center, School of ECEE}\\
  \textit{Arizona State University}\\
Tempe, AZ, USA}
  \and
  \IEEEauthorblockN{Andreas Spanias}
  \IEEEauthorblockA{\textit{SenSIP Center, School of ECEE}\\
  \textit{Arizona State University}\\
Tempe, AZ, USA} }

\maketitle
\setlength{\abovedisplayskip}{2pt}
\setlength{\belowdisplayskip}{2pt}
\setlength{\abovedisplayshortskip}{0pt}
\setlength{\belowdisplayshortskip}{0pt}
\begin{abstract}
Histopathologic cancer detection is challenging due to tissue variability, staining differences, and subtle visual distinctions between disease classes. We propose two quantum algorithms for this task: a configurable dual-gradient CSWAP circuit (DG-CSWAP) that computes multi-directional edge responses in a single execution via per-pixel local $R_y$ encoding, and a hardware-efficient destructive swap circuit (DG-DST) natively matched to quantum processing unit (QPU) gate sets at substantially lower circuit complexity. We prove algebraic equivalence between DG-CSWAP and DG-DST, enabling a two-circuit QPU validation strategy. A three-stage NISQ mitigation pipeline, including readout error correction, bias subtraction, and slope regression, reduces single-pixel hardware MSE by ${\sim}8\times$. Validated on five quantum processors via Amazon Braket, the method achieves inter-platform Pearson $r \approx 0.93$--$0.94$ across all local-simulator pairs. Compared to a prior Quantum Fourier Transform (QFT) based amplitude-encoding baseline requiring 12-qubit global state preparation and a three-model ensemble (85.55\% on PatchCamelyon), the proposed method uses shot-based measurements, executes on real quantum hardware, and achieves 79.80\% accuracy with a single ResNet-50. A Lite configuration delivers a $17\times$ preprocessing speedup at a 2.59\% accuracy cost. To the best of our knowledge, this is the first quantum hardware implementation study with noise mitigation for histopathologic image classification.
\end{abstract}

\begin{IEEEkeywords}
Amazon Braket, Quantum Hardware Implementation, Histopathology Image, Cancer Detection, Quantum Noise Mitigation
\end{IEEEkeywords}

\section{Introduction}

Histopathological cancer detection from Hematoxylin and Eosin (H\&E) stained patches~\cite{madusanka2023} is a grand-challenge image classification problem. The classification process~\cite{touhami2026,dunn2025,lafarge2017} is especially challenging because of high variability in tissue appearance, staining differences, and confusability~\cite{kumar2017} between disease types. Cancerous tissue exhibits subtle micro-architectural changes, such as disorganized gland boundaries, shifted nuclear morphology, and irregular cell-cluster geometry. The problem is further complicated by limited labeled data and image artifacts, which result in significant computational requirements in machine learning training. Classical gradient operators such as Canny respond uniformly to all intensity discontinuities, often amplifying staining artifacts while suppressing diagnostically relevant boundary structure.
\begin{figure}[h]

  \centering
  \includegraphics[width=0.45\textwidth]{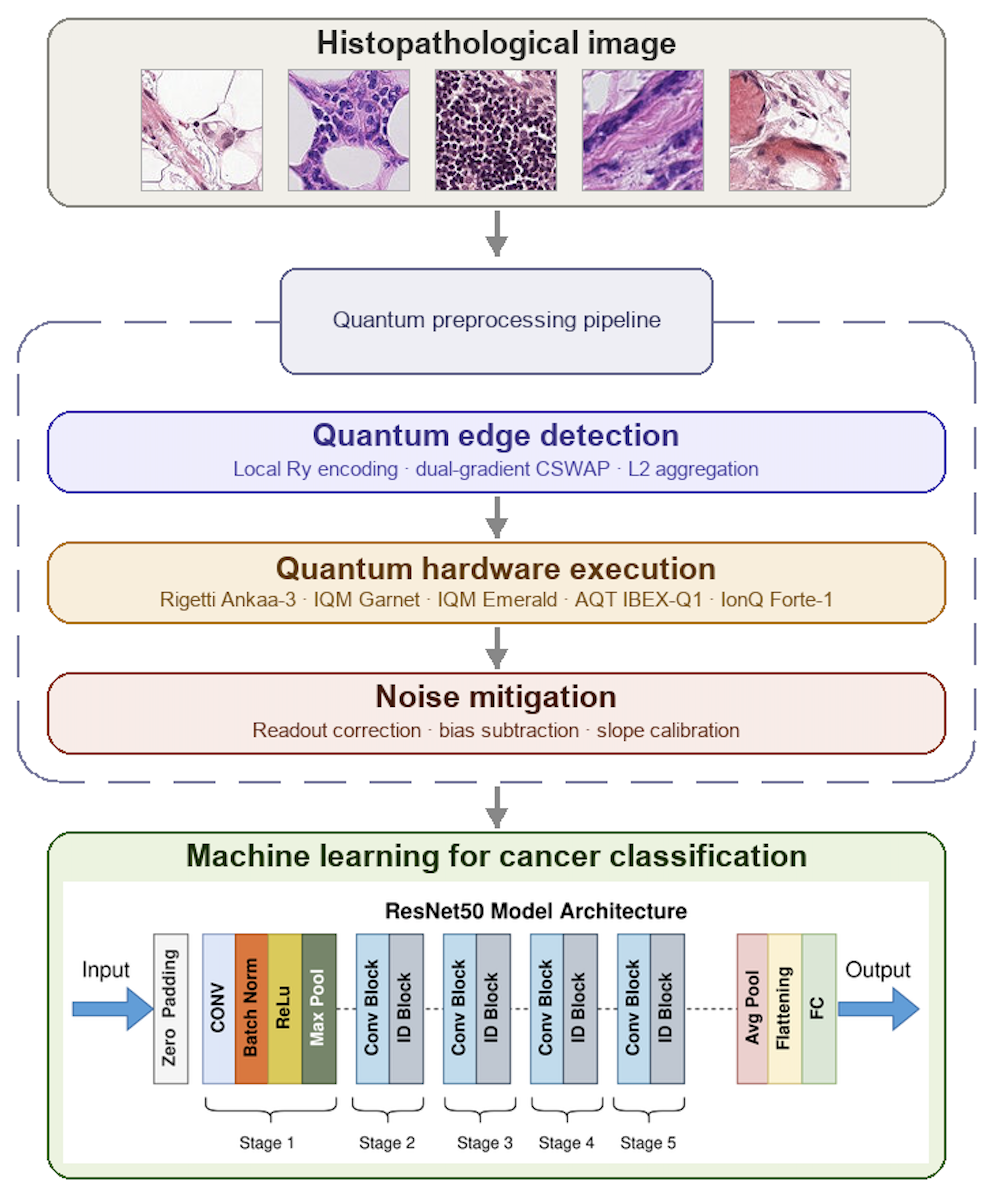}
  \caption{Overview of the proposed quantum edge detection pipeline for
histopathologic cancer classification on real quantum hardware.}
  \label{fig:pipeline}

\end{figure}
In our previous simulation study~\cite{goyal2025}, we observed sensitivity to quantum noise along with challenges related to quantum circuit design and the impact of encoding choices on classification performance. Additional challenges included limited qubit resolution, the requirement to operate strictly within noisy intermediate-scale quantum (NISQ) hardware constraints, and heavy computational overhead in machine learning training with histopathologic image datasets.

In this paper, we propose new quantum hardware-compatible circuit designs along with efficient implementation of new algorithms on quantum hardware for histopathologic cancer classification (Figure~\ref{fig:pipeline}). The SWAP test, originally introduced as a quantum overlap-estimation routine, measures how similar two quantum states are through interference-based measurement~\cite{buhrman2001}. Existing quantum edge detectors encode the full image into a global multi-qubit amplitude state and recover the result through statevector access or small-scale simulation workflows~\cite{yao2017}. That design is difficult to execute on physical QPUs because global state preparation is costly and statevector readout is not available on hardware. Hybrid quantum-classical approaches~\cite{geng2022,nunziata2025,billias2025} reduce circuit depth but have not demonstrated a full image-to-image edge transform on multiple real QPUs. In contrast, our method compares only local neighboring pixel pairs: each pixel is angle-encoded with a single-qubit $R_y$ rotation, and a SWAP-test overlap estimator measures pairwise similarity using finite-shot measurements. Repeating the same local circuit template across the image yields a full edge map without requiring full-state readout. Compared with amplitude encoding, the local approach does not capture image-wide correlations within a single quantum state, but it gains shallow preparation, direct measurability, and execution on real quantum hardware. Our prior work~\cite{goyal2025} was based on a Quantum Fourier Transform (QFT) pipeline and achieved 85.55\% accuracy on PCam with a QFT-based pipeline but remained confined to simulation for this reason.

The proposed methods present a practical measurement protocol for overlap-derived edges that converts per-pixel local $R_y$ encoding into a calibrated, hardware-executable edge-magnitude field, complete with estimator-specific noise correction and QPU validation across five device architectures. Figure~\ref{fig:pipeline} illustrates the proposed end-to-end pipeline, from histopathological image input through quantum edge detection, quantum hardware execution via Amazon Braket, noise mitigation, and final cancer classification.

This study contributes the following to quantum hardware implementation for medical image analysis: (a) we introduce a \emph{dual-gradient overlap primitive}, i.e., a DG-CSWAP circuit that computes horizontal and vertical edge responses simultaneously in a single execution, halving per-pixel circuit overhead compared to sequential evaluation; (b) we replace global amplitude encoding with local single-qubit $R_y$ encoding of neighboring pixel pairs, enabling complete $64{\times}64$ edge-map generation on real QPUs through shot-based overlap measurement without statevector readout; (c) we prove algebraic equivalence between DG-CSWAP and DG-DST, enabling a ``two circuits, one estimator'' QPU validation strategy natively matched to hardware gate sets; and (d) we introduce a three-stage NISQ mitigation pipeline combining readout confusion-matrix inversion, piecewise bias subtraction, and global slope regression, reducing single-pixel hardware MSE by ${\sim}8{\times}$ on a representative superconducting device and applied across all five platforms. Cross-platform full-image validation is performed on five real quantum processors referred to throughout as SC$n$ (superconducting, $n$ qubits) and TI$n$ (trapped-ion, $n$ qubits): SC82, SC20, SC54, TI12, and TI36 (full specifications in Section~\ref{sec:crossplatform}), with PCam classification results reported for both Full and Lite configurations.

The rest of the paper is organized as follows. Section~\ref{sec:method} describes the proposed methods. Section~\ref{sec:experiments} presents current results, and Section~\ref{sec:conclusion} summarizes our conclusions.

\section{Proposed Method}
\label{sec:method}

\subsection{Related Work and Background}
Related work spans quantum edge detection, quantum preprocessing for histopathology, and NISQ error mitigation. 

\textbf{i. Quantum edge detection.}
The SWAP test was introduced by Buhrman et al.~\cite{buhrman2001} for quantum fingerprinting and applied to edge detection as a local overlap measure in the quantum Hadamard edge detection (QHED) framework~\cite{yao2017}. QHED requires encoding the full image as a global $2^n$-qubit amplitude state and reading out the statevector, making it unsuitable for real QPUs. The ancilla-free destructive SWAP variant~\cite{garcia2013} avoids the ancilla qubit. Das et al.~\cite{das2023} validated this method on IBMQ under hardware noise. Hybrid quantum-classical approaches~\cite{geng2022,nunziata2025,billias2025} reduce circuit depth but remain limited to a single device and small-scale validation.

\textbf{ii. Quantum preprocessing for histopathology.}
The PatchCamelyon (PCam) benchmark~\cite{veeling2018} provides $96{\times}96$ Hematoxylin and Eosin (H\&E) patches for patch-level cancer classification. Majumdar et al.~\cite{majumdar2023} apply hybrid quantum-classical transfer learning to PCam at the classification stage rather than preprocessing. Our prior work~\cite{goyal2025} used QFT-based quantum edge extraction as preprocessing, achieving 85.55\% with two ResNet-50 models that included classical and hybrid edge-detection combined via a support vector machine (SVM) meta-classifier. This method required 12-qubit global state preparation and full statevector readout, which confined it to simulation. Our proposed method removes both constraints, i.e., local angle encoding requires one qubit per pixel, and shot-based measurement replaces statevector readout, enabling QPU execution at full image scale.

\textbf{iii. NISQ error mitigation.}
Readout confusion-matrix inversion and bias calibration are standard NISQ mitigation techniques. Das et al.~\cite{das2023} apply them in quantum pattern recognition on real QPUs. The global slope regression stage introduced here addresses visibility attenuation, i.e. the systematic reduction in the measured quantum signal strength relative to the theoretical ideal, specific to the overlap estimator under hardware noise in quantum edge detection.

\subsection{Angle Encoding}

Here we propose encoding each grayscale pixel $p \in [0,1]$ as a single-qubit state via rotation about the Bloch-sphere $y$-axis, $R_y(\theta)=e^{-i\theta Y/2}$:
\begin{equation}
\theta = \operatorname{clip}(g\pi p,\,0,\,\pi), \qquad
|\psi(p)\rangle = R_y(\theta)\,|0\rangle,
\end{equation}
where $g = 1.8$ is a gain factor that increases sensitivity to mid-range intensity differences (clamped to $\pi$ to prevent wrap-around). This encoding maps each pixel to a point on the Bloch sphere, requires exactly one qubit per pixel, and introduces no entanglement in the encoding step.

\subsection{Dual-Gradient Swap-Test Overlap Estimator}

Edge strength between neighboring pixels $p_a$ and $p_b$ is measured as quantum state dissimilarity via the swap test:
\begin{equation}
\begin{aligned}
P(\text{anc}=1) &= \frac{1 - |\langle\psi_a|\psi_b\rangle|^2}{2},\\ e &= 2P(\text{anc}=1) = 1 - |\langle\psi_a|\psi_b\rangle|^2.
\end{aligned}
\label{eq:swap}
\end{equation}
For angle-encoded states, $|\langle\psi_a|\psi_b\rangle|^2 = \cos^2\!\bigl(\tfrac{\theta_a-\theta_b}{2}\bigr)$, so $e = \sin^2\!\bigl(\tfrac{\theta_a-\theta_b}{2}\bigr)$. Identical pixels yield $e=0$. Maximally different pixels yield $e=1$.

The key primitive introduced in this work, the \emph{dual-gradient overlap estimator}, runs two independent swap tests for horizontal (H) and vertical (V) gradients in parallel within a single circuit execution, reducing per-pixel circuit overhead by a factor of two compared to sequential evaluation.

\subsection{6-Qubit DG-CSWAP Circuit}

Two SWAP tests run in parallel within a single 6-qubit circuit, with the following qubit layout:
$q_0,q_1$ (H data), $q_2,q_3$ (V data), $q_4$ (H ancilla), $q_5$ (V ancilla). Each pixel is angle-encoded by an $R_y$ rotation on its corresponding data qubit. The circuit then applies a Hadamard gate to each ancilla, performs the H-channel CSWAP on $(q_4,q_0,q_1)$ and the V-channel CSWAP on $(q_5,q_2,q_3)$ in parallel, applies a final Hadamard gate to each ancilla, and measures all 6 qubits. All qubits are measured to satisfy backends that require full-register measurement. The H and V edge estimates are recovered from the two ancilla measurement probabilities in post-processing.

At the logical level, excluding measurements, Figure~\ref{fig:circuit6} contains ten gates: four $R_y$ gates, four Hadamard gates, and two CSWAP gates. For hardware compilation, each logical CSWAP decomposes into three Toffoli gates. For each channel $c \in \{\mathrm{H},\mathrm{V}\}$, let $P_{1,c}$ denote the measured ancilla-1 probability and let $b_c$ denote the calibrated bias. The corrected edge estimate is
\begin{equation}
e^{(\mathrm{corr})}_c = \operatorname{clip}\!\bigl(2(P_{1,c}-b_c),\,0,\,1\bigr).
\end{equation}

\begin{figure}[h]
\centering
\begin{quantikz}[row sep=0.13cm, column sep=0.35cm]
\lstick{$q_4$} & \gate{H} & \ctrl{1} & \gate{H} & \meter{} \\
\lstick{$q_0$} & \gate{R_y(\theta_{ha})} & \swap{1} & \qw & \qw \\
\lstick{$q_1$} & \gate{R_y(\theta_{hb})} & \targX{} & \qw & \qw \\
\lstick{$q_5$} & \gate{H} & \ctrl{1} & \gate{H} & \meter{} \\
\lstick{$q_2$} & \gate{R_y(\theta_{va})} & \swap{1} & \qw & \qw \\
\lstick{$q_3$} & \gate{R_y(\theta_{vb})} & \targX{} & \qw & \qw \\
\end{quantikz}
\caption{The 6-qubit DG-CSWAP circuit computes H and V edge responses
in parallel. Each ancilla ($q_4$, $q_5$) runs an independent swap test over its data pair. Only the ancilla measurement outcomes are used for edge estimation.}
\label{fig:circuit6}
\end{figure}
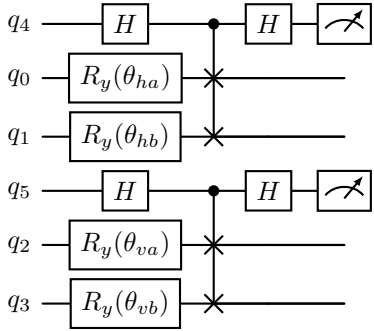

\subsection{Multi-Scale Gradient Evaluation via Pixel-Neighbor Offsets}
\label{sec:multiscale}

The algorithm evaluates edge strength at each pixel by comparing it to its neighbors at spatial offset $d \in \{1, 2\}$ in four directions. Concretely, for pixel $(i,j)$:
\begin{align}
    \text{H:}& \quad \theta_a = \theta_{i,j},\; \theta_b = \theta_{i,j+d}, \nonumber\\
    \text{V:}& \quad \theta_a = \theta_{i,j},\; \theta_b = \theta_{i+d,j}, \nonumber\\
    \text{D1:}& \quad \theta_a = \theta_{i,j},\; \theta_b = \theta_{i+d,j+d}, \nonumber\\
    \text{D2:}& \quad \theta_a = \theta_{i+d,j},\; \theta_b = \theta_{i,j+d}. \nonumber
\end{align}
The output at pixel $(i,j)$ for each direction is the \emph{maximum} edge response over $d \in \{1,2\}$, providing two-scale gradient sensitivity without spatial averaging: $e^{\text{dir}}_{i,j} = \max_{d \in \{1,2\}} e_{i,j}^{\text{dir},d}$. The final per-pixel edge magnitude combines all active directions in L2:
\begin{equation}
m_{i,j} = \sqrt{(e^H_{i,j})^2 + (e^V_{i,j})^2 + \mathbb{I}_{\text{diag}}
    \!\left[(e^{D1}_{i,j})^2 + (e^{D2}_{i,j})^2\right]}.
    \label{eq:mag}
\end{equation}
The map is normalized to $[0,1]$ using the 99th percentile for classification preprocessing or 99.5th percentile for cross-platform validation. The output is a full image-to-image transform where the geometric coverage (directions, diagonals, and offsets) is configurable, in contrast to prior methods that compute a single scalar overlap response per pixel~\cite{yao2017,das2023}.

\section{Current Results}
\label{sec:experiments}

\subsection{Noise Mitigation, Shot Efficiency, and Resource Modeling}
\label{sec:noise}

\subsubsection{Ancilla Bias Calibration}

On real hardware, the ancilla qubit typically exhibits a nonzero probability of being measured in the $|1\rangle$ state even when the two input states are identical. In practice, state-preparation and measurement (SPAM) errors, residual entanglement, readout asymmetry, and coherent over-rotations introduce a systematic offset. To estimate this offset, five calibration circuits are run with equal-state inputs spanning $\theta \in \{0, \pi/4, \pi/2, 3\pi/4, \pi\}$ before processing each image. Because both inputs are identical in each circuit, any nonzero ancilla-1 probability reflects hardware bias rather than edge signal. The bias estimate is $b = \frac{1}{5}\sum_{i=1}^{5} p^{\mathrm{cal}}_i$, and is subtracted from subsequent ancilla-1 probability estimates: $\tilde{p} = p_\mathrm{meas} - b$.

\subsubsection{Readout Correction and Global Slope Calibration}
\label{sec:slope}

The ancilla bias calibration corrects the additive noise floor of the 6-qubit CSWAP estimator. For full-image deployment scripts executed on real QPUs, this complete three-stage correction pipeline is applied to the DG-DST estimator, addressing three distinct error sources in sequence. SC82 full-image runs applied only the first two stages. The slope regression stage was not activated because the 20-replica, 6-pass averaging configuration already suppresses per-replica visibility attenuation through spatial averaging, making the additional calibration overhead unnecessary.

\textbf{Stage 1 --- Readout (ro) correction.}
Each qubit's confusion matrix is measured by preparing all-zeros and all-ones states and recording the resulting bitstring distributions. The $2\times2$ per-qubit matrix is inverted, and the joint confusion matrix for each qubit pair is formed via Kronecker product, recovering the corrected joint probability
\(p_\text{true} = (M_{q_0}^{-1} \otimes M_{q_1}^{-1})\, p_\text{meas}\)
from the measured distribution \(p_\text{meas}\).
This corrects systematic misclassification of $|0\rangle$ and $|1\rangle$ states caused by SPAM errors.

\textbf{Stage 2 --- Piecewise-interpolated bias subtraction.}
Equal-state calibration circuits at $\theta \in \{0, \pi/2, \pi\}$ are run after ro correction. The residual nonzero $p(\Psi^-)$ at each equal-state point represents the hardware bias remaining after ro correction. The bias at any intermediate mean angle $\bar{\theta} = (\theta_a + \theta_b)/2$ is linearly interpolated between the nearest two measurements and subtracted as follows: $\tilde{p} = p_\text{ro} - b(\bar{\theta})$.

\textbf{Stage 3 --- Global slope regression.}
Even after bias subtraction, the hardware response may exhibit a systematic visibility attenuation relative to the ideal $p(\Psi^-)$ curve. Each replica is an independent copy of the circuit assigned to a distinct qubit subset of the device. A global slope $\lambda$ is fitted per replica via least-squares regression through the origin, using known-edge calibration circuits at $(\theta, \delta)$ pairs where $p_\text{ideal}(\delta) = \tfrac{1}{2}\sin^2(\delta/2)$ is analytically known, i.e.,
\begin{equation}
    \lambda = \frac{\sum_i x_i\,(y_i - b_i)}{\sum_i x_i^2}, \quad
x_i = p_\text{ideal}(\delta_i),\; y_i = p_\text{ro}(\theta_i, \delta_i),
\end{equation}
where $b_i = b(\bar{\theta}_i)$ is the interpolated bias at the mean angle of each circuit. The fully corrected estimator is then
\begin{equation}
    \hat{p} = \mathrm{clip}\!\left(\frac{p_\text{ro} - b(\bar{\theta})}{\lambda},\; 0,\; 0.5\right).
    \label{eq:full_correction}
\end{equation}

The Full configuration processes $96\!\times\!96$ patches with offsets $d\in\{1,2\}$, diagonal directions, symmetrization, and an adaptive shot budget of $S_\mathrm{min}{=}300$ to $S_\mathrm{max}{=}2{,}000$. Symmetrization averages edge estimates from both pixel orderings within each pair to reduce directional bias. The Lite configuration reduces to $64\!\times\!64$ patches, offset $d=1$, no diagonal directions, and no symmetrization. Diagonal directions are dropped because H/V gradients capture the dominant boundary structure for patch-level classification, and symmetrization is omitted because both pixel-pair orderings contribute negligibly to classification accuracy while doubling circuit count. Moving from Full to Lite relaxes $\mathrm{SE}_\mathrm{target}$ (the maximum allowed binomial standard error per ancilla probability estimate) by 60\% and reduces circuit count by 18$\times$, yielding a combined 43$\times$ reduction in total shots and 17$\times$ wall-clock speedup on Amazon Braket~\cite{braket}, at a small 2.59\% accuracy cost. If the edge map feeds a downstream classifier, the Lite configuration is strongly preferred.
\subsection{QPU Hardware Validation}
\label{sec:qpu}

\subsubsection{The 4-Qubit DG-DST Circuit}

The 6-qubit CSWAP circuit requires each CSWAP to stay coherent across three Toffoli gates, which is deep for NISQ hardware. For QPU validation we use a shallower
\emph{DG-DST}~\cite{garcia2013} that requires no ancilla.

For a single channel with pixel angles $\theta_a$ (control $c$) and $\theta_b$ (target $t$): (1)~Prepare $R_y(\theta_a)|0\rangle_c$ and $R_y(\theta_b)|0\rangle_t$. (2)~Apply CNOT$(c \to t)$. (3)~Apply H$(c)$. (4)~Measure both qubits; record $P(\text{``}11\text{''}) \equiv P(c{=}1,t{=}1)$. The full circuit for H and V channels uses 4 qubits total and is shown in Figure~\ref{fig:circuit4}. For hardware execution the circuit is compiled to each device's native gate set and submitted via AWS Braket.

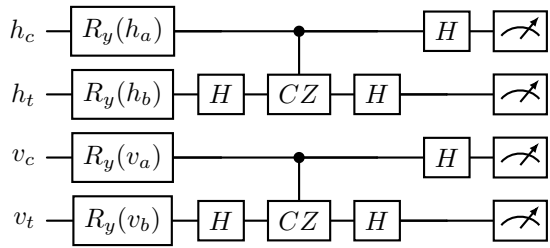
\begin{figure}[h]
\centering
\begin{quantikz}[row sep=0.19cm, column sep=0.32cm]
\lstick{$h_c$} & \gate{R_y(h_a)} & \qw      & \ctrl{1}  & \qw      & \gate{H} & \meter{} \\
\lstick{$h_t$} & \gate{R_y(h_b)} & \gate{H} & \gate{CZ} & \gate{H} & \qw      & \meter{} \\
\lstick{$v_c$} & \gate{R_y(v_a)} & \qw      & \ctrl{1}  & \qw      & \gate{H} & \meter{} \\
\lstick{$v_t$} & \gate{R_y(v_b)} & \gate{H} & \gate{CZ} & \gate{H} & \qw      & \meter{} \\
\end{quantikz}
\caption{The 4-qubit DG-DST circuit. On SC20/SC54 hardware, the CNOT is implemented as $H \cdot CZ \cdot H$ using the CZ gate native to SC20/SC54, and Ry and H gates are compiled to PRx rotations. Bell-basis measurement directly estimates $p(\Psi^{-})$ with no ancilla. Other platforms compile the same logical circuit to their native gate sets (iSWAP, GPI/GPI2/ZZ, RXX) via AWS Braket.}
\label{fig:circuit4}
\end{figure}

\subsubsection{Algebraic Equivalence of DG-CSWAP and DG-DST}

The central validation claim is that both circuits compute the same quantity from hardware-executable gates. This ``two circuits, one estimator'' property means that QPU results on the shallow DG-DST directly validate the edge estimator used for every PCam experiment, since both circuits realize the same overlap measurement.

\begin{proposition}
For pure product inputs $|\psi_a\rangle \otimes |\psi_b\rangle$, DG-CSWAP and DG-DST compute the same quantity, i.e.,
\begin{equation}
e = 1 - |\langle\psi_a|\psi_b\rangle|^2 = \sin^2\!\left(\frac{\theta_a - \theta_b}{2}\right).
    \label{eq:equiv}
\end{equation}
\end{proposition}

\begin{proof}
For the 6-qubit circuit, the swap test gives $P(\text{anc}=1) = \tfrac{1}{2}(1 - |\langle\psi_a|\psi_b\rangle|^2)$, so $e = 2P(\text{anc}=1) = 1 - |\langle\psi_a|\psi_b\rangle|^2$.

For DG-DST, the circuit H$(c)$-CNOT$(c{\to}t)$ is the two-qubit Bell-basis circuit (reversed). It transforms the computational basis as $|b_1 b_0\rangle \mapsto |B_{b_1 b_0}\rangle$, mapping $|11\rangle \mapsto |\Psi^-\rangle$. Therefore:
\begin{equation}
P(\text{``}11\text{''}) = |\langle\Psi^-|\psi_a\otimes\psi_b\rangle|^2 =
    \frac{1 - |\langle\psi_a|\psi_b\rangle|^2}{2},
\end{equation}
so $e = 2P(\text{``}11\text{''}) = 1 - |\langle\psi_a|\psi_b\rangle|^2$. Both expressions equal $\sin^2\!\bigl(\tfrac{\theta_a-\theta_b}{2}\bigr)$ for $R_y$-encoded states. This result holds under the operating assumption of this work that each qubit is independently angle-encoded via $R_y(\theta)|0\rangle$, producing pure product inputs with no inter-register entanglement prior to circuit execution. For correlated inputs $\rho_{ab}$, both circuits generalize identically to $\mathrm{Tr}[(I-\mathrm{SWAP})/2\cdot\rho_{ab}]$, preserving the equivalence while changing the measured quantity~\cite{garcia2013}.
\end{proof}

The two circuits are algebraically identical realizations of the same quantum overlap measurement, chosen for their respective execution contexts: DG-CSWAP is better suited to simulators where measurement flexibility is less constrained, while DG-DST minimizes compiled two-qubit gate count for QPU execution (2 CX vs.\ 36 CX after decomposition, see Table~\ref{tab:complexity}).

\subsubsection{Replica Selection and Error Mitigation}

$R_\mathrm{rep}$ independent replicas of the DG-DST circuit run simultaneously within a single circuit execution, each on a distinct 4-qubit mapping. Qubit assignments are determined per device by scoring all connected qubit pairs using a weighted combination of two-qubit gate fidelity (weight 0.75) and mean readout fidelity of the qubit pair (weight 0.25) from the device calibration snapshot. A maximum-weight disjoint matching is then run over scored pairs, with each replica independently calibrated and fitted, enabling spatial noise profiling across chip regions. Inter-replica $p(\Psi^-)$ spread (0.10--0.15) exceeded the emulator-hardware MSE gap, confirming averaging as the dominant stabilizer. Sentinel circuits confirmed post-correction drift below 0.004, ruling out device drift.

Table~\ref{tab:mse} quantifies the effect of the correction pipeline on SC20. Hardware MSE falls from ${\sim}9\!\times\!10^{-4}$ (raw) to ${\sim}1.1\!\times\!10^{-4}$ (bias-and-scale), which represents an $8\times$ reduction. The emulator--hardware MSE gap ranges from $2.8{\times}$ (raw) to $4.4{\times}$ (readout-corrected), consistent across all mitigation levels (Table~\ref{tab:mse}), confirming the emulator is a reliable development proxy before committing to QPU cost.

\begin{table}[h]
\caption{Mean squared error (MSE) of the corrected $p(\Psi^-)$ estimator against its analytic ideal, measured on SC20 across 10 calibration cases (5 equal-state and 5 known-edge inputs) for both H and V channels. Each row shows the MSE after applying mitigation up to that stage, for hardware (QPU) and emulator separately. HW/Emu is their ratio.}
\label{tab:mse}
\centering\small
\renewcommand{\arraystretch}{0.90}
\begin{tabular}{@{}lccc@{}}
\toprule
\textbf{Mitigation} & \textbf{HW MSE} & \textbf{Emu MSE} & \textbf{HW/Emu} \\
\midrule
Raw          & ${\sim}9\!\times\!10^{-4}$ & ${\sim}3\!\times\!10^{-4}$ & $2.8\times$ \\ Readout (ro) & ${\sim}2\!\times\!10^{-4}$ & ${\sim}5\!\times\!10^{-5}$ & $4.4\times$ \\ Bias+scale   & ${\sim}1.1\!\times\!10^{-4}$ & ${\sim}3\!\times\!10^{-5}$ & $3.6\times$ \\
\bottomrule
\end{tabular}
\end{table}

\subsection{Comparison with AE-QFT Baseline}
\label{sec:comparison}

Table~\ref{tab:novelty} summarizes the key technical differences between the AE-QFT baseline and this work. The 5.75\% accuracy gap (85.55\% ensemble vs.\ 79.80\% single-model) reflects the benefit of AE-QFT's multi-model fusion, not an inherent advantage of amplitude encoding. AE-QFT required two separately trained ResNet-50 models, including a classical baseline and a hybrid edge-detection fusion model, combined via an SVM meta-classifier. This work uses a single ResNet-50. The hybrid Canny overlay peaks at $\alpha=0.2$ (where $\alpha$ is the blend weight of the Canny overlay added to the quantum edge map) with 79.40\% accuracy and degrades beyond $\alpha=0.4$, confirming that quantum edge information should dominate.

\begin{table}[h]
\caption{Technical comparison: AE-QFT baseline vs.\ proposed method.}
\label{tab:novelty}
\centering\footnotesize
\setlength{\tabcolsep}{4pt}
\renewcommand{\arraystretch}{0.90}
\begin{tabular}{@{}p{2cm}p{3.0cm}p{3.0cm}@{}}
\toprule
\textbf{Dimension} & \textbf{AE-QFT Baseline} & \textbf{This Work} \\
\midrule
Encoding & Global amplitude encoding of $64\!\times\!64$ image into 12-qubit statevector & Local angle encoding: one qubit per pixel, $R_y(\theta)$ \\
Core operation & Hadamard + QFT/IQFT basis shifting & Dual-gradient overlap primitive: parallel H/V swap tests in one circuit \\
Qubits & 12 ($\log_2 4096$) & 6 (4 data + 2 ancilla); 4 for QPU DG-DST \\
Output & Statevector readout (simulation only) & Shot-based measurement (QPU-compatible) \\
Multiscale & None & Offsets $d \in \{1,2\}$, 4 directions \\
Mitigation & None & RO correction, piecewise bias subtraction, global slope regression \\
Execution & Qiskit/Aer simulation only & Braket LocalSim + emulators + SC20/SC54 + SC82 QPUs \\
Pipeline & 5 stages, 3 models & 2 stages, 1 model \\
PCam accuracy & 85.55\% (ensemble) & 79.80\% (single model) \\
\bottomrule
\end{tabular}
\end{table}

\subsubsection{Circuit Complexity Analysis}

To quantify the resource gap between the three circuits, we compute standard compiled-gate metrics alongside the \emph{Character Complexity} measure introduced by Shami~\cite{shami}. For a circuit of gates $\{U_g\}$, the fundamental character complexity $\mathrm{CC}_{\mathrm{fund}} = \sum_{g} |\mathrm{Tr}(U_g)|^2$ measures how identity-like each gate is, while the mixing complexity $\mathrm{CC}_{\mathrm{mix}} = \sum_{g} (d_g^2 - |\mathrm{Tr}(U_g)|^2)/d_g^2$, where $d_g$ is the gate dimension, measures how much each gate scrambles the quantum state. $\mathrm{CC}_{\mathrm{mix}}$ is zero for the identity and grows with classical-simulation hardness.
Table~\ref{tab:complexity} summarizes the results. DG-DST (logical basis) is the clear hardware winner: 8~logical gates, depth~4, only 2~two-qubit entangling (CX/CNOT) gates, and the lowest $\mathrm{CC}_{\mathrm{mix}}$ (5.5). On SC20/SC54 hardware each CNOT appears as $H\!\cdot\!CZ\!\cdot\!H$ as shown in Figure~\ref{fig:circuit4}, but the compiled two-qubit gate count remains 2. DG-CSWAP incurs an $18\times$ CX overhead (36 vs.\ 2) because each CSWAP decomposes into three Toffoli gates, each requiring 6~CNOTs and 7~T-gates~\cite{amy}, and two-qubit gates are the dominant error source on real hardware. The compiled total of 162~T-gates consists of 42 from the Toffoli chain ($2{\times}3{\times}7$) plus ${\approx}120$ from Solovay--Kitaev decomposition of the four $R_y$ rotations at $\varepsilon{=}10^{-3}$ (approximation error tolerance; ${\approx}30$ T-gates each)~\cite{dawson2006}. Even at the logical level, its $\mathrm{CC}_{\mathrm{fund}}$ remains $5\times$ larger than DG-DST (80 vs.\ 16), reflecting the presence of two 3-qubit CSWAP operations. The AE-QFT baseline is exponentially more expensive, with $2^{13}{-}2=8{,}190$ CNOTs from state preparation alone~\cite{shende2006} and over 270,000 T-gates under Solovay--Kitaev approximation~\cite{dawson2006} at $\varepsilon=10^{-3}$, placing it beyond near-term execution without quantum RAM.
\begin{table}[h]
\caption{Circuit complexity comparison. Total gates, circuit depth,
$\mathrm{CC}_{\mathrm{fund}}$, $\mathrm{CC}_{\mathrm{mix}}$, and entangling fraction are reported at the logical-circuit level. CX and T-gate counts are reported after compilation/decomposition. ($\theta = \pi/2$, SK precision $\varepsilon=10^{-3}$).}
\label{tab:complexity}
\centering
\footnotesize
\setlength{\tabcolsep}{3pt}
\renewcommand{\arraystretch}{0.90}
\begin{tabular}{@{}l@{\hspace{4pt}}cccc@{}}
\toprule
Metric & DG-CSWAP & DG-DST & DG-DST & QHED \\ &          & SC20/SC54 & logical & 12q \\ &          & native &         &      \\
\midrule
Qubits              & 6    & 4   & 4   & 12 \\ Total gates         & 10   & 26  & 8   & 16{,}718 \\ Circuit depth       & 4    & 8   & 4   & 2{,}073 \\ CX count (compiled) & 36   & 2   & 2   & 8{,}190 \\ T-gate count        & 162  & 720$^\dagger$ & 120 & 270{,}210 \\ $\mathrm{CC}_{\mathrm{fund}}$ & 80.0 & 36.0 & 16.0 & 57{,}657 \\ $\mathrm{CC}_{\mathrm{mix}}$  &  6.9 & 18.5 &  5.5 &  8{,}941 \\ Entangling fraction & 0.20 & 0.08 & 0.25 & 0.49 \\
\bottomrule
\multicolumn{5}{l}{$^\dagger$Native PRX gates require no T decomposition on SC20/SC54 hardware.}
\end{tabular}
\end{table}

\subsection{Cross-Platform Full-Image Validation}
\label{sec:crossplatform}

\subsubsection{Experimental Configuration}
We extended full-image edge-map generation to five platforms available via AWS Braket~\cite{braket}, namely: SC82 (superconducting transmon, native iSWAP gate, 82 qubits), SC20 (superconducting, native CZ gate, 20 qubits), SC54 (superconducting, native CZ gate, 54 qubits), TI36 (trapped-ion, native GPI/GPI2/ZZ gates, 36 qubits), and TI12 (trapped-ion, native $R_{XX}$ M{\o}lmer--S{\o}rensen gate, 12 qubits). All implement the same logical DG-DST estimator $e = 1 - |\langle\psi_a|\psi_b\rangle|^2$ using platform-specific native-gate decompositions. In all cases the logical circuit structure (angle-encode, Bell-basis change, measure) is unchanged, and the algebraic equivalence proved in Section~\ref{sec:qpu} applies without modification.

All runs processed the same $64\!\times\!64$ H\&E histopathology patch (200 shots per image circuit, 1000 shots for calibration). Local simulator and emulator runs for all five platforms (SC82, SC20, SC54, TI36, and TI12) used $R_\mathrm{rep}=3$ replicas and one pass with the Braket LocalSimulator. The SC20 hardware run used $R_\mathrm{rep}=4$, with $P=3$ averaging passes with pixel-order shuffling between passes, and $S_\mathrm{img}=200$, $S_\mathrm{cal}=1000$ shots, with actual cost \$939.68 (124 tasks, 622,400 shots). The SC54 hardware run used $R_\mathrm{rep}=5$ replicas and $P=3$ averaging passes with pixel-order shuffling, $S_\mathrm{img}=200$, $S_\mathrm{cal}=1000$ shots (100 tasks, 500,000 shots; actual cost \$830.00). The SC82 hardware run used $R_\mathrm{rep}=20$ replicas and $P=6$ averaging passes with $S_\mathrm{img}=200$, $S_\mathrm{cal}=1000$ shots (51 tasks, 254,000 shots; actual cost \$243.90). All SC20/SC54, TI36, and TI12 runs applied the full three-stage correction pipeline (ro correction, piecewise bias subtraction, and global slope regression; Section~\ref{sec:slope}). The SC82 hardware, emulator, and local simulator runs applied the two-stage pipeline (ro correction and piecewise bias subtraction only); the slope stage was not activated for those runs. Cross-platform validation uses $(\hat{e}_H{+}\hat{e}_V)/2$ normalized at the 99.5th percentile; PCam preprocessing uses the full L2 magnitude of Eq.~(\ref{eq:mag}) with diagonal directions.

\subsubsection{Evaluation Methodology}

A classical Sobel edge map of the test image serves as the common reference. Two complementary metrics are reported. \emph{Pearson correlation} ($r$) measures how well the spatial pattern of the quantum edge map matches that of the Sobel map, independent of overall brightness, i.e., it captures where edges appear. \emph{ROC AUC} measures whether the quantum map correctly ranks true-edge pixels above non-edge pixels as a binary classification, i.e., it captures how discriminative the edge map is. The top 10\% of Sobel pixels (90th-percentile) are labeled as true edges. AUC of 0.5 is chance, 1.0 is perfect.

An $r\approx0.52$--$0.54$ between the quantum and Sobel maps is expected and correct. The quantum estimator measures the \emph{angle} between qubit-encoded pixel states (geometric dissimilarity), while Sobel measures intensity gradient magnitude (local derivative). Both respond to edges, producing partial agreement. Their divergence in how they weight different pixel pairs is inherent to the algorithms, not a quality failure. The stronger portability claim is inter-platform $r\approx0.93$--$0.94$: all five quantum platforms agree with each other far more closely than with any classical reference, confirming the quantum algorithm extracts a consistent, hardware-agnostic signal independent of device architecture.

\subsubsection{Current Results}

Table~\ref{tab:interplatform} shows inter-platform local-simulator consistency. Table~\ref{tab:crossplatform} summarizes metrics across all configurations. Figure~\ref{fig:cross_platform} compares Pearson $r$ and ROC AUC against the Sobel reference across platforms and backends.

\begin{table}[h]
\caption{Inter-platform local-simulator Pearson $r$. All pairs achieve $r\approx0.93$--$0.94$, indicating highly consistent quantum edge maps across platforms.}
\label{tab:interplatform}
\centering\small
\renewcommand{\arraystretch}{0.90}
\begin{tabular}{@{}lccccc@{}}
\toprule
& SC82 & SC20 & SC54 & TI12 & TI36 \\
\midrule
SC82 & 1.000 & 0.938 & 0.932 & 0.933 & 0.932 \\ SC20      & --    & 1.000 & 0.933 & 0.937 & 0.934 \\ SC54     & --    & --    & 1.000 & 0.932 & 0.931 \\ TI12     & --    & --    & --    & 1.000 & 0.934 \\ TI36    & --    & --    & --    & --    & 1.000 \\
\bottomrule
\end{tabular}
\end{table}

\begin{table}[h]
\caption{Cross-platform edge-map quality vs.\ classical Sobel reference. $r$: Pearson correlation; AUC: threshold-free edge quality (0.5\,=\,chance, 1.0\,=\,perfect). All local sims yield $r\approx0.52$--$0.54$ and AUC$\approx0.80$--$0.81$; SC82 emulator lower due to noise model distortion.}
\label{tab:crossplatform}
\centering\small
\renewcommand{\arraystretch}{0.90}
\begin{tabular}{@{}llcc@{}}
\toprule
\textbf{Platform} & \textbf{Backend} & \boldmath$r$ & \textbf{AUC} \\
\midrule
SC82 & Local Sim & 0.533 & 0.810 \\ SC82 & Emulator  & 0.428 & 0.734 \\ SC82 & Hardware  & 0.458 & 0.764 \\
\midrule
SC20 & Local Sim & 0.536 & 0.804 \\ SC20 & Emulator  & 0.516 & 0.798 \\ SC20 & Hardware  & 0.526 & 0.807 \\
\midrule
SC54  & Local Sim & 0.523 & 0.797 \\
SC54  & Emulator  & 0.518 & 0.803 \\ SC54  & Hardware  & 0.470 & 0.806 \\
\midrule
TI36 & Local Sim & 0.536 & 0.809 \\ TI36 & Emulator  & 0.534 & 0.815 \\
\midrule
TI12  & Local Sim & 0.531 & 0.807 \\ TI12  & Emulator  & 0.516 & 0.806 \\
\bottomrule
\end{tabular}
\end{table}

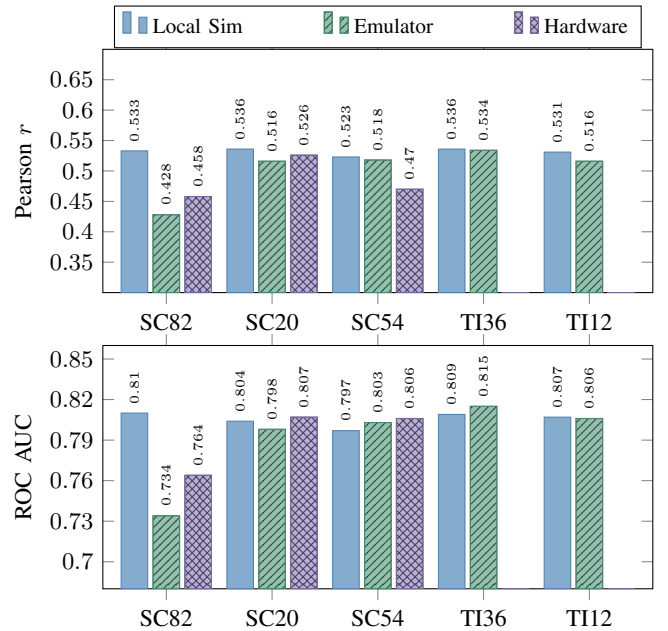
\begin{figure}[h]
\centering
\begin{tikzpicture}
\begin{groupplot}[
group style={ group size=1 by 2, vertical sep=0.7cm, }, ybar=2pt, width=\linewidth, height=4.8cm, enlarge x limits=0.15, xtick=data, xticklabels={SC82, SC20, SC54, TI36, TI12}, xticklabel style={font=\small, rotate=0, anchor=north}, legend style={ at={(0.5,1.18)}, anchor=north, font=\footnotesize, legend columns=3, /tikz/every even column/.append style={column sep=1cm} }, every axis/.append style={font=\small}, ]
\nextgroupplot[
ylabel={Pearson $r$}, ymin=0.30, ymax=0.70, ytick={0.35,0.40,0.45,0.50,0.55,0.60,0.65}, nodes near coords, nodes near coords style={font=\tiny, rotate=90, anchor=west}, every node near coord/.append style={/pgf/number format/fixed, /pgf/number format/precision=3}, ]
\addplot[fill=cBlue!75, draw=cBlue!90!black] coordinates {
(1,0.533)(2,0.536)(3,0.523)(4,0.536)(5,0.531) };
\addplot[fill=cGreen!65, draw=cGreen!85!black,
postaction={pattern=north east lines, pattern color=cGreen!50!black}] coordinates { (1,0.428)(2,0.516)(3,0.518)(4,0.534)(5,0.516) };
\addplot[fill=cPurple!60, draw=cPurple!85!black,
postaction={pattern=crosshatch, pattern color=cPurple!55!black}] coordinates { (1,0.458)(2,0.526)(3,0.470)(4,0)(5,0) };
\legend{Local Sim, Emulator, Hardware}
\nextgroupplot[
ylabel={ROC AUC}, ymin=0.68, ymax=0.86, ytick={0.70,0.73,0.76,0.79,0.82,0.85}, nodes near coords, nodes near coords style={font=\tiny, rotate=90, anchor=west}, every node near coord/.append style={/pgf/number format/fixed, /pgf/number format/precision=3}, ]
\addplot[fill=cBlue!75, draw=cBlue!90!black] coordinates {
(1,0.810)(2,0.804)(3,0.797)(4,0.809)(5,0.807) };
\addplot[fill=cGreen!65, draw=cGreen!85!black,
postaction={pattern=north east lines, pattern color=cGreen!50!black}] coordinates { (1,0.734)(2,0.798)(3,0.803)(4,0.815)(5,0.806) };
\addplot[fill=cPurple!60, draw=cPurple!85!black,
postaction={pattern=crosshatch, pattern color=cPurple!55!black}] coordinates { (1,0.764)(2,0.807)(3,0.806)(4,0)(5,0) };
\end{groupplot}
\end{tikzpicture}
\caption{Preliminary cross-platform quality vs.\ Sobel (single run per platform). All local sims: $r\approx0.52$--$0.54$; SC20 hardware: $r=0.526$, $r(\text{hw,local})=0.931$; SC54 hardware: $r=0.470$, $r(\text{hw,local})=0.895$; SC82 hardware: $r=0.458$. SC82 emulator lower ($r=0.428$) due to noise model distortion. Hardware runs not performed for TI36 or TI12.}
\label{fig:cross_platform}
\end{figure}

\begin{table}[h]
\caption{Local-to-emulator fidelity and mean-edge inflation per platform. $r_{L,E}$: local-sim/emulator Pearson $r$; $\Delta\bar{e}$: mean-edge inflation; $r_{L,H}$: local-sim/hardware $r$ (``---''\,=\,no QPU run performed).}
\label{tab:emulator_summary}
\centering\small
\renewcommand{\arraystretch}{0.90}
\begin{tabular}{@{}lrrl@{}}
\toprule
\textbf{Platform} & \boldmath$r_{L,E}$ & \boldmath$\Delta\bar{e}$ & \boldmath$r_{L,H}$ \\
\midrule
SC82 & 0.776 & $+6.5\%$  & 0.822 \\ SC20      & 0.905 & $+20.0\%$ & 0.931 \\
SC54     & 0.911 & $+9.6\%$  & 0.895 \\ TI12     & 0.921 & $+4.9\%$  & ---   \\ TI36    & 0.921 & $+6.2\%$  & ---   \\
\bottomrule
\end{tabular}
\end{table}

\subsubsection{Discussion}

Five observations follow from Table~\ref{tab:crossplatform}, which are stated below.
\begin{itemize}\setlength{\itemsep}{0pt}\setlength{\parsep}{0pt}\setlength{\topsep}{2pt}

\item \textit{All local simulators agree closely.} All five local simulators
achieve $r\approx0.52$--$0.54$ and AUC$\approx0.80$--$0.81$ against the Sobel reference (Table~\ref{tab:crossplatform}), with inter-platform $r\approx0.93$--$0.94$ (Table~\ref{tab:interplatform}) confirming near-identical quantum edge maps across both hardware families.

\item \textit{SC82 hardware outperforms its own emulator.} The hardware run
achieves $r=0.458$ and AUC$=0.764$ vs.\ the Sobel reference, above the emulator's $r=0.428$ and AUC$=0.734$. The local-to-hardware correlation is $r_{L,H}=0.822$ (Table~\ref{tab:emulator_summary}), compared to local-to-emulator $r_{L,E}=0.776$. This improvement is attributable primarily to the 6-pass, 20-replica averaging applied on hardware, which suppresses shot noise and replica-periodic stripe artifacts beyond what the 1-pass, 3-replica emulator achieves. Per-pass pixel-pair reordering distributes noise load across chip regions, enabling pass averaging to cancel regional variation (spatial chip threading). The SC82 emulator's $r_{L,E}=0.776$ is notably lower than SC20/TI12/TI36 ($r_{L,E}\approx0.89$--$0.92$, Table~\ref{tab:emulator_summary}), suggesting that even the updated SC82 noise model has a somewhat higher structural distortion than the other platforms, not captured by the mean-inflation metric alone.

\item \textit{SC20 hardware: clean full-image result.} The SC20
hardware run ($R_\mathrm{rep}=4$, $P=3$, with pixel-order shuffling between passes) produced an edge map with mean-edge inflation of $+12\%$ relative to the local simulator, and stripe correlation of $-0.001$ (essentially zero), indicating the three-pass shuffle completely suppressed replica-periodic artifacts. Hardware-to-local-simulator Pearson $r=0.931$ (Table~\ref{tab:emulator_summary}) matches the inter-platform local-simulator correlation ($r\approx0.93$), confirming that hardware noise at this operating point introduces no measurable structural degradation beyond the shared quantum signal.

\item \textit{SC54 hardware: consistent with SC20.} The SC54
hardware run ($R_\mathrm{rep}=5$, $P=3$, with pixel-order shuffling) achieved hardware-to-local-simulator Pearson $r_{L,H}=0.895$ and AUC$=0.806$ vs.\ the Sobel reference, with mean-edge hardware inflation of $+23\%$ relative to its local simulator (compared to $+9.6\%$ emulator inflation in Table~\ref{tab:emulator_summary}) and a stripe correlation of $0.016$ substantially suppressed by the three-pass shuffle.

\item \textit{TI12 and TI36 confirm trapped-ion consistency.} TI12 achieves
$r=0.531$ and AUC$=0.807$ on the local simulator, consistent with all other platforms. The emulator shows only 4.9\% mean-edge inflation with $r_{L,E}=0.921$ (Table~\ref{tab:emulator_summary}), confirming the trapped-ion noise model preserves spatial structure. The TI36 local sim ($r=0.536$, AUC$=0.809$) and emulator ($r=0.534$, AUC$=0.815$) differ by less than 0.01 on all metrics, with $r_{L,E}=0.921$ and only $+6.2\%$ mean-edge inflation, confirming the noise model does not significantly alter the edge-map structure at this shot budget.
\end{itemize}

Figure~\ref{fig:roi_comparison} shows representative hardware edge maps on the same H\&E test patch. SC82 and SC20 respond to the same boundary structures across ROIs A--C despite using distinct native gate sets (iSWAP vs.\ CZ). The noisier appearance of the SC82 map reflects its two-stage correction and the higher structural distortion of the SC82 noise model noted in Table~\ref{tab:emulator_summary}. SC20's three-stage pipeline yields a cleaner map ($r_{L,H}=0.931$ vs.\ $r_{L,H}=0.822$). Together, these results confirm the estimator extracts a hardware-agnostic edge signal and establish a measurable baseline from which noise mitigation improvements can be quantified.

\begin{figure}[h]

  \centering
  \includegraphics[width=\linewidth]{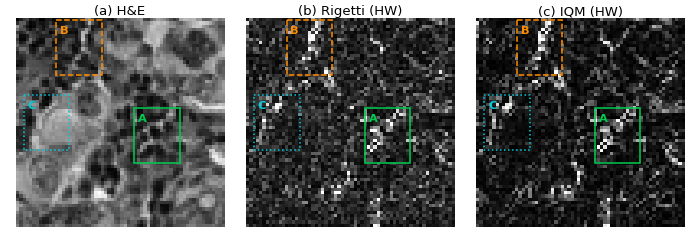}
  \caption{Hardware quantum edge maps on the same $64\!\times\!64$ H\&E patch. (a)~Original. (b)~SC82 ($R_\mathrm{rep}{=}20$, $P{=}6$, two-stage correction). (c)~SC20 ($R_\mathrm{rep}{=}4$, $P{=}3$, three-stage correction). ROIs A--C: A (dense nuclear cluster), B (inter-cellular boundary region), C (gland-lumen interface). Both platforms detect the same boundaries using different native gate sets (iSWAP vs.\ CZ), confirming hardware-agnostic edge detection ($r_{L,H}{=}0.822$ SC82, $r_{L,H}{=}0.931$ SC20).}
  \label{fig:roi_comparison}

\end{figure}

\subsection{Classification Results and Noise Robustness}
\label{sec:results}

\subsubsection{PCam Classification}

All classification experiments use ResNet-50 trained on PCam patches ($96\!\times\!96$), averaged over 5 independent runs. For quantum conditions, the grayscale edge-magnitude map is replicated to three channels; the unprocessed baseline uses full RGB. The DG-CSWAP quantum edge map achieves 79.80\% accuracy and simultaneously raises sensitivity above the no-preprocessing baseline (0.636$\to$0.673) while maintaining specificity well above classical Canny (0.864 vs.\ 0.671), as shown in Figure~\ref{fig:preprocess_comparison}. Braket Full and Qiskit Full configurations agree within 0.47\% (5-run Monte Carlo average), confirming platform-independent behavior. The Lite configuration ($64\!\times\!64$, $d=1$, no diagonals) trades 2.59\% accuracy (79.80\%$\to$77.21\%) for a $17\times$ preprocessing speedup, with Lite achieving the highest sensitivity (0.739) across all conditions (preferred when minimizing false negatives).

\begin{figure}[h]

\centering
\begin{tikzpicture}
\begin{axis}[
ybar=3pt, bar width=10pt, width=\linewidth, height=4.8cm, enlarge x limits=0.18, symbolic x coords={None (baseline), Canny, DG-CSWAP, Hybrid}, xtick=data, xticklabel style={font=\small, rotate=0, anchor=north}, ymin=0.55, ymax=1.02, ytick={0.60,0.70,0.80,0.90,1.00}, yticklabels={0.60,0.70,0.80,0.90,1.00}, ylabel={Score}, ylabel style={font=\small}, tick label style={font=\small}, legend style={font=\small, at={(0.99,0.99)}, anchor=north east}, legend columns=3, grid=major, grid style={dotted,gray!30}, ]
\addplot[fill=cBlue!75, draw=cBlue!90!black] coordinates {
(None (baseline), 0.8213) (Canny,           0.7422) (DG-CSWAP,        0.7980) (Hybrid,          0.7967) };
\addplot[fill=cGreen!65, draw=cGreen!85!black,
postaction={pattern=north east lines, pattern color=cGreen!50!black}] coordinates { (None (baseline), 0.6364) (Canny,           0.8769) (DG-CSWAP,        0.6731) (Hybrid,          0.6827) };
\addplot[fill=cPurple!60, draw=cPurple!85!black,
postaction={pattern=crosshatch, pattern color=cPurple!55!black}] coordinates { (None (baseline), 0.9657) (Canny,           0.6711) (DG-CSWAP,        0.8639) (Hybrid,          0.8569) };
\legend{Val Acc, Sensitivity, Specificity}
\end{axis}
\end{tikzpicture}
\caption{PCam results (5-run MC avg, ResNet-50). Quantum edge map raises
sensitivity (0.636$\to$0.673) while keeping specificity far above Canny (0.864 vs.\ 0.671).}
\label{fig:preprocess_comparison}

\end{figure}
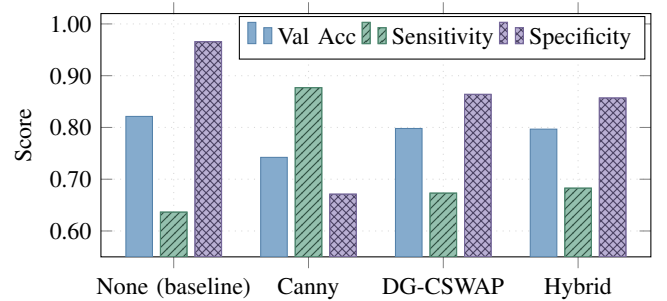

\subsubsection{Noise Robustness}

To assess algorithmic sensitivity to hardware error, we injected three standard quantum error channels (depolarizing, amplitude damping, bit-flip) into the Braket LocalSimulator using the Lite configuration. Across all nine tested configurations (three channels $\times$ three error rates), validation accuracy ranges from 75.75\% to 79.07\%, which is a drop of at most 4.05\% from the noiseless Full result (79.80\%). Depolarizing noise at $p=0.10$ has the least impact (79.07\%), nearly matching the noiseless result. Amplitude damping at $\gamma=0.30$ causes the largest accuracy drop (75.75\%) but yields the highest sensitivity (0.811), suggesting the channel selectively amplifies edge responses relevant to cancer detection. These results confirm the estimator is robust to both real NISQ noise (corrected by the three-stage pipeline) and idealized error models.

\section{Conclusion}
\label{sec:conclusion}

We have introduced a dual-gradient overlap primitive for hardware-compatible quantum edge detection, framed as a practical measurement protocol for overlap-derived edges. The method computes two independent H/V gradient estimates in a single circuit execution, provides a full image-to-image transform with tunable geometry (four directions, two offsets), and is QPU-validated via an algebraically equivalent shallow circuit.
A local angle-encoding approach replaces global amplitude encoding and statevector readout with shot-based QPU-executable measurements, enabling full 64$\times$64 edge-map generation on real hardware. Multi-scale edge evaluation at configurable pixel-neighbor offsets $d \in \{1,2\}$ across four directions is supported. A three-stage NISQ correction pipeline (readout confusion-matrix inversion, piecewise-interpolated bias subtraction, and global slope regression) reduces SC20 single-pixel hardware MSE by ${\sim}8{\times}$. We prove algebraic equivalence between DG-CSWAP and DG-DST for the same estimator $e = 1 - |\langle\psi_a|\psi_b\rangle|^2$, enabling the ``two circuits, one estimator'' QPU validation strategy. Cross-platform full-image validation on SC82, SC20, SC54, TI12, and TI36 yields inter-local-simulator $r\approx0.93$--$0.94$ across all ten pairs, with full-image QPU runs on SC20 and SC54 achieving $r=0.931$ and $r=0.895$ to their local simulators.

Executing these algorithms on real quantum hardware introduced unique challenges, most significantly the high cost of QPU access, which directly limited shot budgets and the number of platforms that could be fully evaluated. The multi-replica, multi-pass averaging strategy adopted in this work is a deliberate response to this constraint. By distributing circuit instances across independent qubit subsets within a single task submission, reliable edge maps are obtained within a fixed cost envelope without scaling per-circuit shot counts.

On PCam, the DG-CSWAP quantum edge map achieves 79.80\% accuracy with a single ResNet-50. It simultaneously raises sensitivity above the unprocessed baseline (0.636$\to$0.673) and maintains high specificity (0.864). The QPU-validated estimator, cross-platform digital-twin characterization across five hardware architectures, and a fully reproducible protocol establish this method as a practical framework for hardware-aware quantum preprocessing in medical image classification. 

We note that the results reported in this paper were obtained in spring of 2026 and reflect the capabilities of the hardware, compilers, and cloud platforms available at that time. As quantum computing hardware and software ecosystems are rapidly evolving, updated results will be reported at the conference and in the final paper.


\section*{Acknowledgment}
This work was supported by the SenSIP Center at Arizona State University. Amazon Web Services provided access to Amazon Braket and quantum hardware. Some of the tools developed in this work may also benefit an ongoing National Science Foundation (NSF IUSE) education project.

\enlargethispage{2\baselineskip}
\balance

\end{document}